\documentclass[aps,pra,twocolumn,superscriptaddress]{revtex4-1}
\usepackage{braket}
\usepackage{amsmath,amssymb,amsfonts}
\usepackage{amsthm}
\usepackage{graphicx}

\newtheorem{theorem}{Theorem}[section]
\newtheorem{lemma}[theorem]{Lemma}
\newtheorem{prop}[theorem]{Proposition}

\newcommand{\average}[1]{\mbox{$\langle #1 \rangle$}}

\newcommand{\Tr}[1]{\mbox{$\mathrm{Tr}[#1]$}}

\newcommand{\be}{\begin{equation}}
\newcommand{\ee}{\end{equation}}
\newcommand{\ba}{\begin{eqnarray}}
\newcommand{\ea}{\end{eqnarray}}
\newcommand{\ban}{\begin{eqnarray*}}
\newcommand{\ean}{\end{eqnarray*}}

\begin{document}

\title{Robust self testing of the 3-qubit $W$ state}
\author{Xingyao Wu}
\affiliation{Centre for Quantum Technologies, National University of Singapore, 3 Science Drive 2, 117543, Singapore}
\author{Yu Cai}
\affiliation{Centre for Quantum Technologies, National University of Singapore, 3 Science Drive 2, 117543, Singapore}
\author{Tzyh Haur Yang}
\affiliation{Centre for Quantum Technologies, National University of Singapore, 3 Science Drive 2, 117543, Singapore}
\author{Huy Nguyen Le}
\affiliation{Centre for Quantum Technologies, National University of Singapore, 3 Science Drive 2, 117543, Singapore}
\author{Jean-Daniel Bancal}
\affiliation{Centre for Quantum Technologies, National University of Singapore, 3 Science Drive 2, 117543, Singapore}
\author{Valerio Scarani}
\affiliation{Centre for Quantum Technologies, National University of Singapore, 3 Science Drive 2, 117543, Singapore}
\affiliation{Department of Physics, National University of Singapore, 2 Science Drive 3, 117542, Singapore}

\begin{abstract}
Self-testing is a device independent method which can be used to determine the nature of a physical system or device, without knowing any detail of the inner mechanism or the physical dimension of Hilbert space of the system. The only information required are the number of measurements, number of outputs of each measurement and the statistics of each measurement. Earlier works on self testing restricted either to two parties scenario or multipartite graph states. Here, we construct a method to self-test the three-qubit $W$ state, and show how to extend it to other pure three-qubit states. Our bounds are robust against the inevitable experimental errors.

\end{abstract}

\begin{widetext}
\maketitle
\end{widetext}

\section{Introduction}

The certification of a quantum state is an important step in many quantum information tasks. When the device is provided by an untrusted vendor, it is beneficial to certify the device under minimal assumptions. Device-independent protocols are developed for this purpose. The device-independent approach describes the experiment from the observed statistics of measurement results, without any assumptions the inner workings of the devices and the description of the physical system, other than the fact that quantum theory holds. The possibility of such a certification is the applied offspring of Bell's theorem: by sharing entangled quantum states, distant parties could establish correlation stronger than those obtained by shared randomness only. Furthermore, some specific statistics require the parties to share a particular quantum state, up to local isometries (because certainly such a blind assessment cannot characterize local unitaries, nor the presence of other degrees of freedom that are not measured). Self-testing refers to such cases.

The word comes from the pioneering paper of Mayers and Yao in 2004 \cite{MayersYao2004} that described a criterion for the self-testing of the maximally entangled state of two qubits. The self-testing of the same state had been demonstrated previously \cite{sw87,pr92,Tsirelson93} using the maximal violation of the CHSH inequality \cite{CHSH}. These proofs were simplified and made robust to small deviations from the ideal case~\cite{McKague2012Robust}. Recently it was noticed that the maximally entangled state of two qubits is not the only one that can be self-tested: McKague extended self-testing to all multipartite graph states~\cite{matthew}, Yang and Navascu\'es to any pure bipartite entangled state~\cite{Yang2013Robust}. Self-testing has also been extended to deal with parallel repetition scenarios \cite{miller,vazirani}. Recently, the robustness of several of these results have been imporved and new states, including qutrit states, have been self-tested~\cite{prl}. In this paper, we present the self-testing of the $W$ state of three qubits~\cite{Duer2000Three}
\ba
\ket{W_3}&=&\frac{1}{\sqrt{3}}(\ket{001}+\ket{010}+\ket{100})\label{eq:W3}
\ea and of other pure 3-qubit states which are not graph states.

This paper is organized as follows. In section~\ref{sec:w_state} we first give a review of the 2-qubit self-testing, and then we show how to self test the $W_3$ state. Section~\ref{sec:robustness} is devoted to show the method we proposed for self-testing $W_3$ state is robust under the experimental statistics fluctuations. In section~\ref{sec:general}, we extend our method to more general 3-qubit states.

\section{Self-testing of {\bf the} $W_{3}$ state}
\label{sec:w_state}

\subsection{Mayers-Yao test for two qubits}

Firstly, we introduce the notions of self-testing by reviewing the Mayers-Yao test for the maximally entangled state of two qubits, which will also be useful below. Two devices, each of which allegedly performs non-commuting measurement of a qubit belonging to a maximally entangled pair, are given to Alice and Bob, who treat them as black boxes: they can only query them with several possible settings; to any query, the box produces an outcome. We consider the case where the outcomes are all binary. Confident that quantum theory is a correct description of physics, Alice and Bob can assign an unknown quantum state to the content of the boxes. In fact, they can assume the state to be pure by considering that the boxes include, if needed, purifying degrees of freedom (this is not an adversarial scenario, in which a third party may benefit from holding a purification). Since the number of degrees of freedom under study is not bounded, they can assume that the measurements on the state are projective. This means that for any setting of Alice there exist an otherwise unknown pair of projectors $P_{M_A=\pm 1}$, and the same for Bob. Alice and Bob assume that the boxes behave locally, i.e. $[P_{M_A=a},P_{N_B=b}]=0$ for all measurements $(M_A,N_B)$ and for all outcomes. This is all that can be said \textit{a priori}; the rest of the evidence is constituted by the observation \textit{a posteriori} of the statistics produced by querying the boxes.

At this point, we need to define precisely what self-testing means, since the mapping of classical statistics to quantum system is one-to-many. We say that an unknown quantum state $\ket{\psi}_{AB}$ is self-tested into a well-defined two qubit state $\ket{\phi}$ if there exist a local isometry $\Phi=\Phi_{A}\otimes\Phi_{B}$ such that
\begin{align}
\Phi\ket{\psi}_{AB}\ket{00}_{A^{'}B^{'}}&=\ket{\textrm{junk}}_{AB}\ket{\phi}_{A^{'}B^{'}} \nonumber\\
\Phi M_{A}N_{B}\ket{\psi}_{AB}\ket{00}_{A^{'}B^{'}}&=\ket{\textrm{junk}}_{AB}(\sigma_{m}\otimes\sigma_{n}\ket{\phi^{+}}_{A^{'}B^{'}}), \label{definition}
\end{align}
where $\ket{00}_{A^{'}B^{'}}$ is a product state ancilla attached by Alice and Bob locally into their system for the sake of {\bf the} argument. The $\sigma$-s here are the Pauli matrices acting on the ancilla qubit space. Thus, the unknown state $\ket{\psi}_{AB}$ and the unkown measurements $\{M_A,N_B\}$ can be mapped to the two qubit state $\ket{\phi}$ and suitable Pauli matrices respectively. The local isometry takes care of the degeneracy of the problem with respect to local unitaries and the addition of irrelevant degrees of freedom.

It is crucial to understand that the isometry is a \textit{virtual} protocol, not one to be implemented in the lab (the requirement of having trusted qubits being at odds with the device-independent nature of the task). The performance of the unknown boxes in this local isometry has to be assessed at the mathematical level, on the basis of the observed statistics: these statistics come from \textit{direct query} of the boxes, as explained above. For instance, instead of working with the projectors, it will be convenient to work with unitary hermitian operators of the form $M_A=P_{M_A=+1}-P_{M_A=-1}$: expressions like $M_A^2=\mathbb{I}$ are obvious properties of those mathematical objects and have nothing to do with assuming something on repeated measurements.
 
The Mayers-Yao test involves two dichotomic measurements $(X_A,Z_A)$ for Alice and three dichotomic measurements $(X_B,Z_B,D_B)$ for Bob. If the following statistics of the measurements are observed:
\begin{align}
\bra{\varPsi}Z_{A}Z_{B}\ket{\varPsi}=\bra{\varPsi}X_{A}X_{B}\ket{\varPsi}=1\,,\\
\bra{\varPsi}X_{A}Z_{B}\ket{\varPsi}=\bra{\varPsi}Z_{A}X_{B}\ket{\varPsi}=0\,,\\
\bra{\varPsi}Z_{A}D_{B}\ket{\varPsi}=\bra{\varPsi}X_{A}D_{B}\ket{\varPsi}=\frac{1}{\sqrt{2}}\,,
\end{align}
then indeed the content of the boxes can be self-tested into the maximally entangled state of two qubits and the suitable complementary Pauli matrices \cite{MayersYao2004,McKaguethesis}.

\subsection{Criterion for self-testing of $W_3$}\label{sectionIIb}

Consider now that Alice, Bob and Charlie share a quantum state, with the local measurements $M_A$, $N_B$ and $L_C$ respectively. To be able to self test the $W$ state \eqref{eq:W3}, we need to identify statistics that witness the state uniquely in the sense of (\ref{definition}): there exist a local isometry $\Phi = \Phi_A \otimes \Phi_B \otimes \Phi_C$ such that 
\begin{align}
\label{isometry}
&\Phi \left(M_A N_B L_C \ket{\varPsi}_{ABC}\ket{000}_{A'B'C'}\right) \nonumber\\
&=\ket{\textrm{junk}}_{ABC}( \sigma_m \otimes \sigma_n \otimes \sigma_l \ket{W_3}_{A'B'C'}). 
\end{align}
Our construction is based on a reduction to the Mayers-Yao two-qubit self-testing. Each $\Phi_Q$, $Q=A,B,C$, is the same that was used for the Mayers-Yao test \cite{McKague2012Robust}. The ancilla qubit of each party is initialized to $\ket{0}$, so that initial state of the whole system, including the device and the ancillae is $\ket{\varPsi}_{ABC}\ket{000}_{A'B'C'}$. The whole circuit of the isometry is represented in Fig.~\ref{fig:isometry}.

\begin{figure}[htbp!]
\centering
\includegraphics[width=0.4\textwidth]{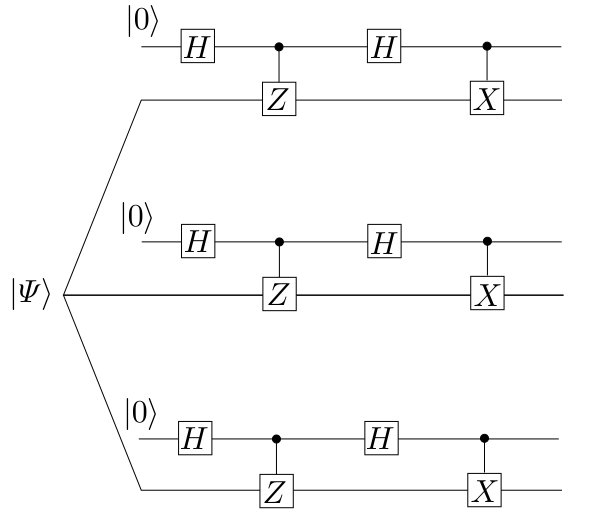}
\caption{The local isometry used to self-test $W_3$ state. The $X$ and $Z$ are the unitary operators defined by the projective measurements of the box, via $M_A=P_{M_A=+1}-P_{M_A=-1}$. The $H$ here are the standard Hadamard gate while the control-$X$ and control-$Z$ apply $X$ and $Z$ respectively when the control qubit is in the state $\ket{1}$.}
\label{fig:isometry}
\end{figure}

Now we can proceed to prove the self-testing criterion:
\begin{theorem}
\label{lemma0}
Alice, Bob and Charlie, spatially separated, each performs two measurements with binary outcomes on an unknown shared quantum state $\ket{\Psi}$. The $W_3$ state is self-tested if the following statistics are observed:
\begin{align}
\label{require1}
&\bra{\varPsi}P_{A}^{0}P_{B}^{0}P_{C}^{1}\ket{\varPsi}=\bra{\varPsi}P_{A}^{0}P_{B}^{1}P_{C}^{0}\ket{\varPsi}=\bra{\varPsi}P_{A}^{1}P_{B}^{0}P_{C}^{0}\ket{\varPsi}\nonumber\\&=\dfrac{1}{3}\\
&\begin{cases}
\label{require2}
\bra{\varPsi}P_{A}^{0}X_{B}X_{C}\ket{\varPsi}=-\bra{\varPsi}P_{A}^{0}Z_{B}Z_{C}\ket{\varPsi}=\dfrac{2}{3}\\
\bra{\varPsi}P_{A}^{0}X_{B}D_{C}\ket{\varPsi}=-\bra{\varPsi}P_{A}^{0}Z_{B}D_{C}\ket{\varPsi}=\dfrac{1}{\sqrt{2}}\dfrac{2}{3}\\
\bra{\varPsi}P_{A}^{0}X_{B}Z_{C}\ket{\varPsi}=0
\end{cases}\\
&\begin{cases}
\label{require3}
\bra{\varPsi}P_{B}^{0}X_{A}X_{C}\ket{\varPsi}=-\bra{\varPsi}P_{B}^{0}Z_{A}Z_{C}\ket{\varPsi}=\dfrac{2}{3}\\
\bra{\varPsi}P_{B}^{0}X_{A}D_{C}\ket{\varPsi}=-\bra{\varPsi}P_{B}^{0}Z_{A}D_{C}\ket{\varPsi}=\dfrac{1}{\sqrt{2}}\dfrac{2}{3}\\
\bra{\varPsi}P_{B}^{0}X_{A}Z_{C}\ket{\varPsi}=0
\end{cases}
\end{align}
\normalsize
where ${P^0}\equiv P_{Z=+1}=\frac{1+Z}{2}$ and ${P^1}\equiv P_{Z=-1}=\frac{1-Z}{2}$ are projectors for the $Z$ measurement.
\end{theorem}

\begin{proof}

The isometry in Figure \ref{fig:isometry} gives as an output
\begin{align}
&\ket{\Psi'}=\Phi\ket{\varPsi}_{ABC}\ket{000}_{A'B'C'} \nonumber\\
=&\dfrac{1}{8}[
(1+Z_{A})(1+Z_{B})(1+Z_{C})\ket{\varPsi}\ket{000}\nonumber\\
&+(1+Z_{A})(1+Z_{B})X_{C}(1-Z_{C})\ket{\varPsi}\ket{001}\nonumber\\
&+(1+Z_{A})X_{B}(1-Z_{B})(1+Z_{C})\ket{\varPsi}\ket{010}\nonumber\\
&+(1+Z_{A})X_{B}(1-Z_{B})X_{C}(1-Z_{C})\ket{\varPsi}\ket{011}\nonumber\\
&+X_{A}(1-Z_{A})(1+Z_{B})(1+Z_{C})\ket{\varPsi}\ket{100}\nonumber\\
&+X_{A}(1-Z_{A})(1+Z_{B})X_{C}(1-Z_{C})\ket{\varPsi}\ket{101}\nonumber\\
&+X_{A}(1-Z_{A})X_{B}(1-Z_{B})(1+Z_{C})\ket{\varPsi}\ket{110}\nonumber\\
&+X_{A}(1-Z_{A})X_{B}(1-Z_{B})X_{C}(1-Z_{C})\ket{\varPsi}\ket{111}] \label{general}\\
&=\,\sum_{a,b,c\in\{0,1\}} X_A^aX_B^bX_C^c\,P_A^aP_B^bP_C^c \ket{\varPsi}\ket{abc}\,.
\end{align}
Observation (\ref{require1}) implies that $\bra{\varPsi}P_{A}^{0}P_{B}^{0}P_{C}^{1}\ket{\varPsi}+\bra{\varPsi}P_{A}^{0}P_{B}^{1}P_{C}^{0}\ket{\varPsi}+\bra{\varPsi}P_{A}^{1}P_{B}^{0}P_{C}^{0}\ket{\varPsi}=1$: therefore $P_{A}^{a}P_{B}^{b}P_{C}^{c}\ket{\varPsi}=0$ for the other five projectors.

From the fact that $\langle\varphi\ket{\psi}=1$ implies $\ket{\varphi}=\ket{\psi}$ and ${P_A^0}^2=P_A^0$, observation (\ref{require2}) implies after some manipulation
\begin{align}
P_{A}^{0}X_{B}\ket{\varPsi}&=P_{A}^{0}X_{C}\ket{\varPsi},\nonumber\\
P_{A}^{0}Z_{B}\ket{\varPsi}&=-P_{A}^{0}Z_{C}\ket{\varPsi},\nonumber\\
P_{A}^{0}X_{B}\ket{\varPsi}&\perp P_{A}^{0}Z_{B}\ket{\varPsi},\nonumber\\
P_{A}^{0}D_{C}\ket{\varPsi}&=\dfrac{P_{A}^{0}X_{B}-P_{A}^{0}Z_{B}}{\sqrt{2}}\ket{\varPsi},\nonumber\\
&=\dfrac{P_{A}^{0}X_{C}+P_{A}^{0}Z_{C}}{\sqrt{2}}\ket{\varPsi}. \label{proof:eqns1}
\end{align} 
Similarly, from (\ref{require3}), we obtain the following 
\begin{align}
P_{B}^{0}X_{A}\ket{\varPsi}&=P_{B}^{0}X_{C}\ket{\varPsi},\nonumber\\
P_{B}^{0}Z_{A}\ket{\varPsi}&=-P_{B}^{0}Z_{C}\ket{\varPsi},\nonumber\\
P_{B}^{0}X_{A}\ket{\varPsi}&\perp P_{B}^{0}Z_{C}\ket{\varPsi},\nonumber\\
P_{B}^{0}D_{C}\ket{\varPsi}&=\dfrac{P_{B}^{0}X_{A}-P_{B}^{0}Z_{A}}{\sqrt{2}}\ket{\varPsi},\nonumber\\
&=\dfrac{P_{B}^{0}X_{C}+P_{B}^{0}Z_{C}}{\sqrt{2}}\ket{\varPsi}\label{proof:eqns2}
\end{align}
The last identity in (\ref{proof:eqns1}) implies
\begin{align}
&(P_{A}^{0}D_{C})^{2}\ket{\varPsi}=P_{A}^{0}D_{C}^{2}\ket{\varPsi}\nonumber\\
&=P_{A}^{0}\ket{\varPsi}=P_{A}^{0}(X_{B}-Z_{B})^{2}\ket{\varPsi}.
\end{align}
Since $X_{B}^{2}=Z_{B}^{2}=D_{C}^{2}=1$, we have
\begin{align}
P_{A}^{0}X_{B}Z_{B}\ket{\varPsi}=-P_{A}^{0}Z_{B}X_{B}\ket{\varPsi},\label{xzb}
\end{align}
and similarly,
\begin{align}
\label{anticom2}
&P_{A}^{0}X_{C}Z_{C}\ket{\varPsi}=-P_{A}^{0}Z_{C}X_{C}\ket{\varPsi}\\
\label{anticom3}
&P_{B}^{0}X_{A}Z_{A}\ket{\varPsi}=-P_{B}^{0}Z_{A}X_{A}\ket{\varPsi}\\
\label{anticom4}
&P_{B}^{0}X_{C}Z_{C}\ket{\varPsi}=-P_{B}^{0}Z_{C}X_{C}\ket{\varPsi}.
\end{align}

All these properties of the operators deduced from the measurement requirements will help to reduce the general output (\ref{general}) to
\begin{align}
\label{aim}
\ket{\tilde{\Psi}}=P_{A}^{0}P_{B}^{0}P_{C}^{0}X_{C}\ket{\varPsi}(\ket{001}+\ket{010}+\ket{100}),
\end{align}
see the proof in Appendix~\ref{app:isometry}. This state can be normalized into the form of $\ket{\textrm{junk}}_{ABC}\ket{W_{3}}_{A'B'C'}$ . Thus we have proven that with these requirements, (\ref{require1}), (\ref{require2}) and (\ref{require3}) on the measurement results indeed self test the unknown state as a $W_{3}$ state.

This completes the self-testing of the $W_{3}$ state. As for self-testing the measurements, it can be shown that (\ref{isometry}) holds indeed by putting operators in front of (\ref{general}) and going through the same steps.
\end{proof}

\section{Robustness}
\label{sec:robustness}
Interesting as the above result is in itself, it relies on observing the measurement statistics in (\ref{require1}), (\ref{require2}) and (\ref{require3}) exactly, which is not possible due to inevitable experimental uncertainties. 

Suppose each observation in (\ref{require1}), (\ref{require2}) and (\ref{require3}) has a deviation at most equal to $\varepsilon$ around the perfect value. For instance, for the first 3 observations,

\begin{align}
\begin{cases}
\label{erequire1}
&|\bra{\varPsi}P_{A}^{0}P_{B}^{0}P_{C}^{1}\ket{\varPsi}-1/3|\leq\varepsilon\\
&|\bra{\varPsi}P_{A}^{0}P_{B}^{1}P_{C}^{0}\ket{\varPsi}-1/3|\leq\varepsilon\\
&|\bra{\varPsi}P_{A}^{1}P_{B}^{0}P_{C}^{0}\ket{\varPsi}-1/3|\leq\varepsilon\\
\end{cases}.
\end{align}

We present two approaches to robustness: the first one (subsection \ref{ss:a}) is based on the analytic method first proposed in~\cite{McKague2012Robust}; the second one (subsection \ref{ss:swap}) uses techniques based on semi-definite programming, following \cite{prl,swap}. Both approaches are converted to fidelity for comparison. As we are going to see, the second method gives a much higher robustness.

\subsection{Analytic bound on the norm}
\label{ss:a}

Assuming the experiment statistics deviate from (\ref{require1}), (\ref{require2}), (\ref{require3}) by a small $\varepsilon$, the self-testing is robust if the isometry still extract a state close to the $W_3$ state in the sense
\begin{align}
\label{eq:robust}
||\ket{\Psi'}-\ket{\textrm{junk}}_{ABC}\ket{W_{3}}_{A'B'C'}||\leq f(\varepsilon),
\end{align}   
where $f(\varepsilon) \rightarrow 0$ when $\varepsilon \rightarrow 0$. 

We show in Appendix \ref{sec:detailcal} that:
\begin{align}
\label{analytical}
\|\ket{\Psi'}&-\ket{\textrm{junk}}_{ABC}\ket{W_{3}}_{A'B'C'}\|\nonumber\\
&\leq  7.5\varepsilon+119.2\varepsilon^{\frac{3}{4}}+49.4\varepsilon^{\frac{1}{4}}.
\end{align}
The proof is based on the argument that if the observation is close to the ideal, the properties of the operators must be close to the ideal operators. Hence, the state extracted using the isometry is close to the ideal one.

\subsection{Bound on the fidelity using semi-definite programs}
\label{ss:swap}
The second method to study robustness follows the technique presented in~\cite{swap} and, rather than using the 2-norm, uses the fidelity of the state (\ref{general}) on the ancillary qubits with the $W_3$ state. This fidelity can be expressed in term of expectation values. Explicitly:
\begin{align}\label{eq:fidel}
F=&\|\bra{W_3^{A'B'C'}}\Phi\ket{\varPsi}_{ABC}\ket{000}_{A'B'C'}\| \nonumber\\
=&\|\frac{1}{8\sqrt{3}}\{(1+Z_{A})(1+Z_{B})X_{C}(1-Z_{C})\ket{\varPsi} \nonumber\\
&+(1+Z_{A})X_{B}(1-Z_{B})(1+Z_{C})\ket{\varPsi} \nonumber\\
&+X_{A}(1-Z_{A})(1+Z_{B})(1+Z_{C})\ket{\varPsi}\}\| \nonumber\\
=&\sum_{i} \alpha_{i}\braket{M_A^i N_A^i O_A^i M_B^i N_B^i O_B^i M_C^i N_C^i O_C^i},
\end{align}
where $M, N, O \in \{1, Z, X\}$ and $\alpha_i$'s are appropriate coefficients.
Some of these average values can be measured, so they can be replaced by the observed values in the expression above. Those that contain both $X$ and $Z$ for a same qubit cannot be measured. Nevertheless, they are not unconstrained: for instance, in the ideal case we know that quantum mechanics necessarily implies tight equalities like $P_A^0(Z_BX_BZ_B)=-P_A^0X_B$ [equation \eqref{xzb}]. Relaxed versions of such constraints must still hold in the non-ideal case. In order to explain how we are going to implement these relaxed constraints, we start by stating the following Lemma \cite{miguel2007}:
\begin{lemma}
\label{positive}
Let \{$A_1,...,A_n$\} be a collection of operators. Then for any quantum state $\rho$ the orthogonal matrix M,
\begin{align}
(M)_{ij}=\Tr{\rho A_i^\dagger A_j},
\end{align}
is non-negative. 
\end{lemma}
Explicitly, any matrix $M$ containing products of our operators \{$Z_A,X_A,Z_B,X_B,Z_C,X_C,D_C$\} with their adjoints, for instance
\scriptsize
\[ M = \left(  \begin{array}{ccccc}
1 & Z_A & Z_AX_A & X_B & Z_BX_C\\
Z_A & 1 & X_A & Z_AX_B & Z_AZ_BX_C\\
X_AZ_A & X_A & 1 & X_AZ_AX_B& X_AZ_AZ_BX_C\\
X_B & Z_AX_B & Z_AX_AX_B & 1 & X_BZ_BX_C\\
Z_BX_C & Z_AZ_BX_C & Z_AX_AZ_BX_C & Z_BX_BX_C & 1 
 \end{array} \right)\]
\normalsize
must be positive semidefinite when evaluated on any quantum state.

Now we would like to find the minimal possible value of $F$ compatible with the relaxed constraints (\ref{require1}), (\ref{require2}), (\ref{require3}) with deviations $\varepsilon_i$s and with quantum physics. The latter condition is equivalent to requiring the matrix $M$ built with \textit{all} the products of our operators to be positive semidefinite \cite{miguel2007}. Such a matrix is obviously infinite, thus impossible to use in practice. By requiring the positivity of a finite submatrix, however, one obtains a relaxation of the constraints. Now the minimization has become a semi-definite program (SDP):
\begin{equation*}
  \begin{array}{lrcll}
    \min 
    & F \\\\
    \text{s.t.}
    &(\ref{require1}), &(\ref{require2}), &(\ref{require3}) \text{ with errors } \varepsilon_i\text{s}\\
     & and & M\geq 0\\
  \end{array}
\end{equation*}
The SDP leads to a valid lower bound on the fidelity for two reasons: first, by choosing a particular finite $M$, we are minimizing over a larger set (fewer constraints) than the set of quantum correlations, whence the quantum minimum can only be higher; second, there is no guarantee that the isometry we started with is actually optimal for this task.

Note that even though the expression~\eqref{eq:fidel} does contain any moment involving the measurement $D$, its appearance in the matrix $M$ makes it useful to bound the fidelity.

In order to find a good bound, and in particular recover the perfect case $F=1$ when (\ref{require1})-(\ref{require3}) hold, the matrix $M$ must be large enough to contain at least all the average values $\braket{\cdot}$ that appear in the expression of $F$. Figure \ref{fig:swapW} shows the result of the SDP optimization for two choices of $M$. The higher dimension we choose, the more detailed the matrix should is, hence, the tighter the bound.


\begin{figure}[htbp!]
\centering
\includegraphics[width=0.48\textwidth]{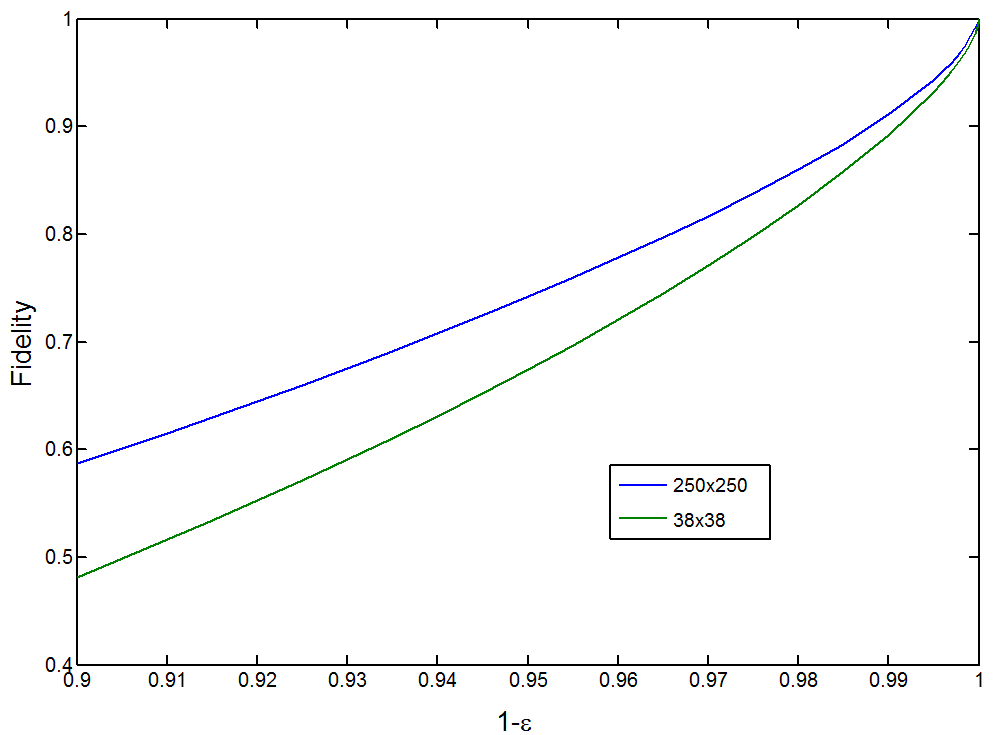}
\caption{Swap bound on the fidelity of the $W_3$ state for different matrix size of $M$. $250 \times 250$ represents the bound given by matrix $M$ of size $250$ and $38 \times 38$ represents the bound given by matrix $M$ of size $38$. The largest matrix corresponds to a relaxation of the NPA hierarchy~\cite{npa} at local level 2 (i.e. it includes any products with at most two operators per party)~\cite{Moroder}. If the fidelity is below $66.7\%$, it may not possess tripartite entanglement.}
\label{fig:swapW}
\end{figure}

\section{More general three qubit states}
\label{sec:general}
In the previous sections, we have shown how one can self test a $W_3$ state. This together with previous result on self testing of GHZ state \cite{matthew} shows that both representatives of the two inequivalent LOCC classes of three qubits~\cite{Dur} can be self-tested. The question then remains whether one can self-test every pure 3-qubits state. Here we explicitly shows how one can self-test a large family of 3 qubit states using bipartite inequalities.

\subsection{Reminder: self-testing of any pure 2-qubit state}

Firstly, let us review how the self testing of arbitrary qubit pairs works. It has been shown that any pure two qubit state in their Schmidt form
\begin{eqnarray}
\ket{\psi_\gamma} = \frac{1}{\sqrt{1+\gamma^2}} \left( \ket{01} + \gamma \ket{10} \right), 
\end{eqnarray}
can be self tested by observing the maximum violation of the tilted CHSH inequality~\cite{Yang2013Robust}: 
\begin{eqnarray}
\label{eqn:tilted}
\beta\left( \alpha, A_0,A_1,B_0,B_1 \right) &=& \alpha A_0 + A_0(B_0+B_1)\\ \nonumber
 &&+ A_1(B_0 - B_1) \leq \textbf{2+$\alpha$},
\end{eqnarray}
where $\alpha=\frac{2\gamma}{1-\gamma^2}$. Note that for simplicity, we have used the notations that $A_0,A_1,B_0$ and $B_1$ to represent the unknown measurements by Alice and Bob respectively.

The maximal quantum violation of this inequality is given by $\beta^* = \sqrt{8+2\alpha^2}$~\cite{acin2012randomness}, achievable with the following measurement settings 
\begin{eqnarray}
\label{eqn:titled-measurement}
A_0 &=& \sigma_z, \nonumber\\ 
A_1 &=& \sigma_x, \nonumber\\
B_0 &=& \cos\mu \sigma_z + \sin\mu \sigma_x, \nonumber\\ 
B_1 &=& \cos\mu \sigma_z - \sin\mu \sigma_x, 
\end{eqnarray}
where $\tan\mu = \frac{2\gamma}{1+\gamma^2}$.

\subsection{Self-testing of a family of pure 3-qubit states}

Now we use the same idea as in Section~\ref{sectionIIb} for the $W_3$ case: base the self-testing of 3-qubit states on a suitable chaining of two self-testing procedures for 2-qubit states. The class of 3 qubit states that can be self-tested with this approach is
\begin{eqnarray}
\label{eqn:general_three_qubit}
\ket{\psi_\gamma} = \frac{1}{\sqrt{2+\gamma^2}}\left(\ket{100} + \ket{010} + \gamma \ket{001} \right)_{ABC},
\nonumber\\
\end{eqnarray}
where $\gamma$ is a real number and $\gamma \neq 0$. 

Notice that the state is symmetric with respect to party A and B. For instance, if we partition the parties into A\textbar BC, we have
\begin{eqnarray}
\ket{\psi} = \frac{1}{\sqrt{2+\gamma^2}} \left(\ket{100}_{ABC}
+ \ket{0}_A ( \underbrace{\ket{10} + \gamma\ket{01}}_{\equiv\ket{\psi}_{BC}} )_{BC} \right). \nonumber\\ 
\end{eqnarray}
On the other hand, if we partition into B\textbar AC, we have
\begin{eqnarray}
\ket{\psi} = \frac{1}{\sqrt{2+\gamma^2}} \left( \ket{100}_{BAC}
+ \ket{0}_B ( \underbrace{\ket{10} + \gamma\ket{01}}_{\equiv\ket{\psi}_{AC}} )_{AC} \right). \nonumber\\
\end{eqnarray}
In either form above, the state $\ket{\psi}_{BC}$ (\textit{resp.} $\ket{\psi}_{AC}$) conditioned on the outcome "0" in the measurement in the $Z$ basis of A (\textit{resp.} B), violates the tilted CHSH inequality maximally. Measurements are set according to \eqref{eqn:titled-measurement}, with $\tan\mu = \frac{2\gamma}{1+\gamma^2}$. Note that Charlie performs the same measurement regardless of which partition we consider, due to the symmetry in the state. 

To sum it up: 
\begin{lemma}
\label{lemma2}{\bf  [Proof in Appendix~\ref{general3qubitproof}]}
Given three black boxes with two buttons each, labelled as $A_0, A_1, B_0, B_1, C_0, C_1$, the following statistics:
\begin{center}
\begin{eqnarray}
\average{P_A^1P_B^0} = \average{P_A^0P_B^1} = \average{P_A^0P_B^0}/\gamma^2 =  \frac{1}{2+\gamma^2}, \nonumber \\
\average{P_A^0 \beta(\alpha,B_0,B_1,C_0,C_1)} = \frac{\beta^*(1+\gamma^2)}{2+\gamma^2}, \nonumber \\
\average{P_B^0 \beta(\alpha,A_0,A_1,C_0,C_1)} = \frac{\beta^*(1+\gamma^2)}{2+\gamma^2},
\end{eqnarray}
\end{center}
where $\beta$ refers to a Bell expression as in~\eqref{eqn:tilted}, $\gamma,\beta^*$ some real number with $\beta^* = \sqrt{8+2\alpha^2}$, and $\alpha = \frac{2\gamma}{1-\gamma^2}$, constitute a self testing of a quantum state of the form \eqref{eqn:general_three_qubit}.
\end{lemma}


General 3-qubit pure state can be written in a standard form with four amplitudes and one phase~\cite{acin2000}. The state we presented above only produces a one parameter class of states, hence it does not cover all the 3-qubit pure states. For more general 3-qubit states that lack the symmetry we used in different partitions, our approach does not apply straightforwardly. In that case, specific states can still always be self-tested using the method we used in section \ref{ss:swap}, for any guess of the measurements and the isometry (see~\cite{swap}).

\section{Conclusion}
In this paper, we have presented a procedure to self-test the $W_3$ state. This procedure makes use of the self-testing of two-qubit states and is robust to small errors. This method generalizes directly to $W$ state of more than 3 qubits.

Our approach of selftesting tripartite states by combining bipartite self-testing schemes allowed us to also self-test a continuous one-parameter family \eqref{eqn:general_three_qubit} of 3-qubit pure states.

The robust bounds that we obtained by SDP constitute the first demonstration of an application of the self-testing technique presented in~\cite{swap} to multipartite states. In agreement with bipartite studies, the SDP method provides better robustness than the analytical bounds. For instance, the fidelity remains above $90\%$ whever the error is bounded by $1\%$.

Note that although the SDP method is tighter than analytical one, the size of the SDP matrix can grow quickly as the number of the parties increases. Hence it would be useful to incorporate symmetric properties into the SDP matrix in order to make the resolution of the problem easier. This remains work for the future.

\section*{Note added}

While completing this work, we became aware of another approach to self-testing many-qubit states, where the explicit application to the $W_3$ state is discussed \cite{otherw3}.

\section*{Acknowledgments}
We acknowledge interesting discussions with Melvyn Ho. This work was supported by the National Research Foundation (partly through the Academic Research Fund Tier 3 MOE2012-T3-1-009) and the Ministry of Education, Singapore.

\appendix

\section{Proof of Eq. (\ref{aim})}
\label{app:isometry}
In this section, we show that the output of the isometry depicted in Fig.~\ref{fig:isometry} is indeed extracts a $W$ state whenever the observed statistics satisfy Eq.(\ref{require1}-\ref{require3}).
Recall that the output of the isometry on any general state is,
\begin{align}
\ket{\Psi'}&=\Phi\ket{\varPsi}_{ABC}\ket{000}_{A'B'C'} \nonumber\\
&=\sum_{a,b,c\in\{0,1\}} X_A^aX_B^bX_C^c\,P_A^aP_B^bP_C^c\ket{\varPsi}\ket{abc}.
\end{align}
All $P_{A}^{a}P_{B}^{b}P_{C}^{c}\ket{\varPsi}=0$ except $abc$ $\in$ $\{001,010,100\}$, hence
\begin{align}
\ket{\Psi'}=
&P_A^0 P_B^0 X_C P_C^1 \ket{\varPsi}\ket{001}\nonumber\\
+&P_A^0 X_B P_B^1 P_C^0 \ket{\varPsi}\ket{010}\nonumber\\
+&X_A P_A^1 P_B^0 P_C^0 \ket{\varPsi}\ket{100}.\nonumber
\end{align}
Using the anti-commuting relations~(\ref{xzb}-\ref{anticom4}), we can deduce that,
\begin{align}
\ket{\Psi'} &= P_A^0P_B^0P_C^0X_C\ket{\varPsi}\ket{001} \nonumber  \\
&+ P_A^0P_B^0P_C^0X_B \ket{\varPsi}\ket{010} \nonumber \\
&+ P_A^0P_B^0P_C^0X_A \ket{\varPsi}\ket{100} \nonumber.
\end{align}
Notice that $P_B^0X_A\ket{\varPsi}=P_B^0X_C\ket{\varPsi}$ and $P_A^0X_B\ket{\varPsi}=P_A^0X_C\ket{\varPsi}$, so

\begin{align}
\ket{\Psi'} = P_A^0P_B^0P_C^0X_C \ket{\varPsi} (\ket{001}+\ket{010}+\ket{100}).
\end{align}
This completes the proof.

\section{Detailed calculation of analytic bound on the norm}
\label{sec:detailcal}

This appendix provides the details of the derivation of the analytical bound (\ref{analytical}).

We first introduce the following lemma:
\begin{prop}
\label{lemma1}
Suppose each statistics in (\ref{require1}), (\ref{require2}) and (\ref{require3}), in the order of appearance, has a deviation $\varepsilon_1, \varepsilon_2,...,\varepsilon_{13}$ from its expected value, e.g.
\begin{align}
\bra{\varPsi}P_{A}^{0}P_{B}^{0}P_{C}^{1}\ket{\varPsi}=\frac{1}{3}+\varepsilon_1,
\end{align}
for the first term, then,
\begin{align}
\begin{cases}
\label{lemma1.1}
\|(P_{A}^{0}X_{B}Z_{B}+P_{A}^{0}Z_{B}X_{B})\ket{\varPsi}\|\leq\delta_{1}\\
\|(P_{A}^{0}X_{C}Z_{C}+P_{A}^{0}Z_{C}X_{C})\ket{\varPsi}\|\leq\delta_{2}\\
\|(P_{A}^{0}X_{B}-P_{A}^{0}X_{C})\ket{\varPsi}\|\leq\delta_{3}\\
\|(P_{A}^{0}Z_{B}+P_{A}^{0}Z_{C})\ket{\varPsi}\|\leq\delta_{4},\\
\end{cases}\\
\begin{cases}
\label{lemma1.2}
\|(P_{B}^{0}X_{A}Z_{A}+P_{B}^{0}Z_{A}X_{A})\ket{\varPsi}\|\leq\delta_{5}\\
\|(P_{B}^{0}X_{C}Z_{C}+P_{B}^{0}Z_{C}X_{C})\ket{\varPsi}\|\leq\delta_{6}\\
\|(P_{B}^{0}X_{A}-P_{B}^{0}X_{C})\ket{\varPsi}\|\leq\delta_{7}\\
\|(P_{B}^{0}Z_{A}+P_{B}^{0}Z_{C})\ket{\varPsi}\|\leq\delta_{8},\\
\end{cases}
\end{align} 
where $\delta_{i}$s are functions of $\varepsilon_i$s. 
\end{prop}

\begin{proof}
We give the proof for (\ref{lemma1.1}), the proof for (\ref{lemma1.2}) is similar.

To be rigorous, we assume
\begin{align}
&\bra{\varPsi}P_{A}^{0}P_{B}^{0}P_{C}^{0}\ket{\varPsi} = \varepsilon_{14}\nonumber\\ &\bra{\varPsi}P_{A}^{0}P_{B}^{1}P_{C}^{1}\ket{\varPsi} = \varepsilon_{15}\nonumber\\
&\bra{\varPsi}P_{A}^{1}P_{B}^{0}P_{C}^{1}\ket{\varPsi} = \varepsilon_{16}\nonumber\\
&\bra{\varPsi}P_{A}^{1}P_{B}^{1}P_{C}^{0}\ket{\varPsi} = \varepsilon_{17}\nonumber\\
&\bra{\varPsi}P_{A}^{1}P_{B}^{1}P_{C}^{1}\ket{\varPsi} = \varepsilon_{18}\nonumber.\\
\end{align}
We can now write
\begin{align}
&\|P_{A}^{0}X_{B}\ket{\varPsi}\|=\sqrt{|\bra{\varPsi}P_{A}^{0}X_{B}X_{B}P_{A}^{0}\ket{\varPsi}|}\nonumber\\
&=\sqrt{|\bra{\varPsi}(P_{A}^{0})^{2}\ket{\varPsi}|}=\sqrt{|\bra{\varPsi}P_{A}^{0}\ket{\varPsi}|}\nonumber\\
&=\sqrt{\frac{2}{3}-(\varepsilon_1+\varepsilon_2+\varepsilon_{14}+\varepsilon_{15})}\nonumber\\
&=\sqrt{\frac{2}{3}-\delta_0}.
\end{align}
where $\delta_0=\varepsilon_1+\varepsilon_2+\varepsilon_{14}+\varepsilon_{15}$ and they all come from results which involve $P_A^0$.
Similarly,
\begin{align}
&\|P_{A}^{0}Z_{B}\ket{\varPsi}\|=\|P_{A}^{0}X_{C}\ket{\varPsi}\|=\|P_{A}^{0}Z_{C}\ket{\varPsi}\|\nonumber\\
&=\|P_{A}^{0}D_{C}\ket{\varPsi}\|=\sqrt{\frac{2}{3}-\delta_0}.
\end{align}
Then,
\begin{align}
&\|(P_{A}^{0}X_{B}-P_{A}^{0}X_{C})\ket{\varPsi}\|\nonumber\\
&=\sqrt{|\bra{\varPsi}P_{A}^{0}X_{B}X_{B}P_{A}^{0}+P_{A}^{0}X_{C}X_{C}P_{A}^{0}-2P_{A}^{0}X_{B}X_{C}P_{A}^{0}\ket{\varPsi}|}\nonumber\\
&=\sqrt{
\begin{aligned}
&|\bra{\varPsi}P_{A}^{0}X_{B}X_{B}P_{A}^{0}\ket{\varPsi}+\bra{\varPsi}P_{A}^{0}X_{C}X_{C}P_{A}^{0}\ket{\varPsi}\\
&-2\bra{\varPsi}P_{A}^{0}X_{B}X_{C}P_{A}^{0}\ket{\varPsi}|
\end{aligned}
}\nonumber\\
&=\sqrt{|(\frac{2}{3}-\delta_0)\times 2-2\times (\frac{2}{3}-\varepsilon_4)|}\nonumber\\
&=\sqrt{2|\delta_0-\varepsilon_4|}.
\end{align}
Using the same techniques, we are able to get
\begin{align}
\|(P_{A}^{0}Z_{B}-P_{A}^{0}Z_{C})\ket{\varPsi}\|=\sqrt{2|\delta_0-\varepsilon_4|}.
\end{align}
To get the first line of (\ref{lemma1.1}), we estimate the following distance,
\begin{align}
\label{ap1}
&\|(P_{A}^{0}D_{C}-\frac{P_{A}^{0}X_{B}-P_{A}^{0}Z_{B}}{\sqrt{2}})\ket{\varPsi}\|\nonumber\\
&=\sqrt{
\begin{aligned}
&|\bra{\varPsi}P_{A}^{0}D_{C}D_{C}P_{A}^{0}+\frac{1}{2}P_{A}^{0}X_{B}X_{B}P_{A}^{0}\nonumber\\
&+\frac{1}{2}P_{A}^{0}Z_{B}Z_{B}P_{A}^{0}-P_{A}^{0}X_{B}Z_{B}P_{A}^{0}\nonumber\\
&-\sqrt{2}P_{A}^{0}D_{C}X_{B}P_{A}^{0}+\sqrt{2}P_{A}^{0}D_{C}Z_{B}P_{A}^{0}\ket{\varPsi}|
\end{aligned}
}\nonumber\\
&=\sqrt{
\begin{aligned}
&|(\frac{2}{3}-\delta_0)\times 2-\sqrt{2}\times(\frac{1}{\sqrt{2}}\frac{2}{3}-\varepsilon_{6})\nonumber\\
&-\sqrt{2}\times(\frac{1}{\sqrt{2}}\frac{2}{3}-\varepsilon_{7})-\bra{\varPsi}P_{A}^{0}X_{B}Z_{B}\ket{\varPsi}|
\end{aligned}
}\nonumber\\
&=\sqrt{|\sqrt{2}(\sqrt{2}\delta_0-\varepsilon_{6}-\varepsilon_7)+\bra{\varPsi}P_{A}^{0}X_{B}Z_{B}\ket{\varPsi}|}\nonumber\\
&\leq\sqrt{
\begin{aligned}
&|\sqrt{2}(\sqrt{2}\delta_0-\varepsilon_{6}-\varepsilon_7)|+|\varepsilon_{6}|\nonumber\\
&+\sqrt{\frac{2}{3}-\delta_0}\times\sqrt{|2(\delta_0-\varepsilon_{4})|}\nonumber\\
\end{aligned}}\\
&=\dfrac{\delta_{1}}{(2+2\sqrt{2})\sqrt{2}},
\end{align}
where the $|\bra{\varPsi}P_{A}^{0}X_{B}Z_{B}\ket{\varPsi}|$ is estimated by using the triangle inequality $|a+b|\leq|a|+|b|$ and Cauchy–--Schwarz inequality $|a\cdot b|\leq|a|\cdot|b|$,
\begin{align}
&|\bra{\varPsi}P_{A}^{0}X_{B}Z_{B}\ket{\varPsi}|=|\bra{\varPsi}P_{A}^{0}X_{B}(Z_{C}+Z_{B}-Z_{C})\ket{\varPsi}|\nonumber\\
&\leq|\bra{\varPsi}P_{A}^{0}X_{B}Z_{C}\ket{\varPsi}|+|\bra{\varPsi}P_{A}^{0}X_{B}(Z_{B}-Z_{C})\ket{\varPsi}|\nonumber\\
&=|\varepsilon_{6}|+|\bra{\varPsi}P_{A}^{0}X_{B}(Z_{B}-Z_{C})\ket{\varPsi}|\nonumber\\
&\leq|\varepsilon_{6}|+\|P_{A}^{0}X_{B}\ket{\varPsi}\|\cdot\|(Z_{B}-Z_{C})\ket{\varPsi}\|\nonumber\\
&=|\varepsilon_{6}|+\sqrt{\frac{2}{3}-\delta_0}\times\sqrt{|2(\delta_0-\varepsilon_{4})|}.
\end{align}
In order to estimate the $\|(P_{A}^{0}X_{B}Z_{B}+P_{A}^{0}Z_{B}X_{B})\ket{\varPsi}\|$, first consider,
\begin{align}
&(P_{A}^{0}D_{C})^{2}\ket{\varPsi}\nonumber\\
&=(\frac{P_{A}^{0}X_{B}-P_{A}^{0}Z_{B}}{\sqrt{2}}+P_{A}^{0}D_{C}-\frac{P_{A}^{0}X_{B}-P_{A}^{0}Z_{B}}{\sqrt{2}})^{2}\ket{\varPsi}\nonumber\\
&=(\frac{P_{A}^{0}X_{B}-P_{A}^{0}Z_{B}}{\sqrt{2}})^{2}\ket{\varPsi}\nonumber\\
&+(P_{A}^{0}D_{C}-\frac{P_{A}^{0}X_{B}-P_{A}^{0}Z_{B}}{\sqrt{2}})^{2}\ket{\varPsi}\nonumber\\
&+2(\frac{P_{A}^{0}X_{B}-P_{A}^{0}Z_{B}}{\sqrt{2}})(P_{A}^{0}D_{C}-\frac{P_{A}^{0}X_{B}-P_{A}^{0}Z_{B}}{\sqrt{2}})\ket{\varPsi}.
\end{align}
The first term contains the anticommutative terms while the last two terms have a same factor. By using the identity $D_{C}^{2}\ket{\varPsi}=X_{B}^{2}\ket{\varPsi}=Z_{B}^{2}\ket{\varPsi}=\ket{\varPsi}$, we can easily deduce that,
\begin{align}
&\frac{(P_{A}^{0}X_{B}Z_{B}+P_{A}^{0}Z_{B}X_{B})\ket{\varPsi}}{\sqrt{2}}\nonumber\\
=&(P_{A}^{0}D_{C}+\frac{P_{A}^{0}X_{B}-P_{A}^{0}Z_{B}}{\sqrt{2}})\nonumber\\
&(P_{A}^{0}D_{C}-\frac{P_{A}^{0}X_{B}-P_{A}^{0}Z_{B}}{\sqrt{2}})\ket{\varPsi},
\end{align}
and the norm can be estimated,
\begin{align}
&\|\frac{(P_{A}^{0}X_{B}Z_{B}+P_{A}^{0}Z_{B}X_{B})\ket{\varPsi}}{\sqrt{2}}\|\nonumber\\
&\leq\|P_{A}^{0}D_{C}(P_{A}^{0}D_{C}-\frac{P_{A}^{0}X_{B}-P_{A}^{0}Z_{B}}{\sqrt{2}})\ket{\varPsi}\|\nonumber\\
&+\frac{1}{\sqrt{2}}\|P_{A}^{0}X_{B}(P_{A}^{0}D_{C}-\frac{P_{A}^{0}X_{B}-P_{A}^{0}Z_{B}}{\sqrt{2}})\ket{\varPsi}\|\nonumber\\
&+\frac{1}{\sqrt{2}}\|P_{A}^{0}Z_{B}(P_{A}^{0}D_{C}-\frac{P_{A}^{0}X_{B}-P_{A}^{0}Z_{B}}{\sqrt{2}})\ket{\varPsi}\|\nonumber\\
&\leq(\|P_{A}^{0}D_{C}\|_{\infty}+\frac{1}{\sqrt{2}}\|P_{A}^{0}X_{B}\|_{\infty}+\frac{1}{\sqrt{2}}\|P_{A}^{0}Z_{B}\|_{\infty})\nonumber\\
&\times\|(P_{A}^{0}D_{C}-\frac{P_{A}^{0}X_{B}-P_{A}^{0}Z_{B}}{\sqrt{2}})\ket{\varPsi}\|.
\end{align}
The infinite norm can be estimated
\begin{align}
&\|P_{A}^{0}D_{C}\|_{\infty}\leq\|P_{A}^{0}\|_{\infty}\|D_{C}\|_{\infty}\nonumber\\
&=\|P_{A}^{0}\|_{\infty}\|P_{D}^{0}-P_{D}^{1}\|_{\infty}\leq\|P_{A}^{0}\|_{\infty}(\|P_{D}^{0}\|_{\infty}+\|P_{D}^{1}\|_{\infty})\nonumber\\
&=2.
\end{align}
Similar for $\|P_{A}^{0}Z_{B}\|_{\infty}$ and $\|P_{A}^{0}X_{B}\|_{\infty}$. Then, we get
\begin{align}
\|\frac{(P_{A}^{0}X_{B}Z_{B}+P_{A}^{0}Z_{B}X_{B})\ket{\varPsi}}{\sqrt{2}}\|\leq \dfrac{\delta_{1}}{\sqrt{2}}.
\end{align}
The other relations can be proved similarly.
\end{proof}

Using the above proposition, we can now turn into the robustness of the $W_{3}$ state. We shall still use the same isometry as described in Figure (\ref{fig:isometry}), irrespective of the errors in the statistics. The output state can always be displayed as (\ref{general}), the problem then is whether we can prove it's close to the target state $\ket{\textrm{junk}}_{ABC}\ket{W_{3}}_{A'B'C'}$. However, it's not easy to figure out what is the exact form of this target state.
What we want to do first is to estimate the distance
\begin{align}
\|\ket{\Psi'}-\ket{\tilde{\Psi}}\|,
\end{align}
where the $\ket{\Psi'}$ is given in (\ref{general}) and $\ket{\tilde{\Psi}}$ is given in (\ref{aim}) as,
\begin{align}
\ket{\tilde{\Psi}}=P_{A}^{0}P_{B}^{0}P_{C}^{0}X_{C}\ket{\varPsi}(\ket{001}+\ket{010}+\ket{100}).
\end{align}
There would be 8 terms regarding to the ancillary qubits. We need to estimate the norm of each term. Since there are too many terms involved, we shall only show explicitly some of them, for instance the term
\begin{align}
&\|\frac{1}{8}(1+Z_{A})(1+Z_{B})X_{C}(1-Z_{C})\ket{\varPsi}\ket{001}\nonumber\\
&-P_{A}^{0}P_{B}^{0}P_{C}^{0}X_{C}\ket{\varPsi}\ket{001}\|\nonumber\\
&=\|P_{A}^{0}P_{B}^{0}X_{C}P_{C}^{1}\ket{\varPsi}\ket{001}-P_{A}^{0}P_{B}^{0}P_{C}^{0}X_{C}\ket{\varPsi}\ket{001}\|.
\end{align}

From Proposition \ref{lemma1}, we could see that the operations $Z$ and $X$ are almost anticommutative. Thus, the operators $Z$ and $X$ in (\ref{general}) other than $\ket{001}$, $\ket{010}$ and $\ket{100}$ can always be moved to the right of $(1-Z)$ with the cost of a small error. Using (\ref{lemma1.2}), we have,
\begin{align}
&\|P_{A}^{0}P_{B}^{0}X_{C}P_{C}^{1}\ket{\varPsi}\ket{001}-P_{A}^{0}P_{B}^{0}P_{C}^{0}X_{C}\ket{\varPsi}\ket{001}\|\nonumber\\
=&\|P_{A}^{0}P_{B}^{0}X_{C}Z_{C}\ket{\varPsi}\ket{001}+P_{A}^{0}P_{B}^{0}Z_{C}X_{C}\ket{\varPsi}\ket{001}\|\nonumber\\
\leq &\|P_{B}^{0}\|_{\infty}\|P_{A}^{0}X_{C}Z_{C}\ket{\varPsi}\ket{001}+P_{A}^{0}Z_{C}X_{C}\ket{\varPsi}\ket{001}\|\nonumber\\
=&\delta_{1}.
\end{align}
Similarly, the terms with $\ket{010}$ and $\ket{100}$ can also be shown to be bounded by the same errors. For the other 5 terms in (\ref{general}), by using the properties $X^2=1$, it shows that,
\begin{align}
&\|P_{A}^{0}P_{B}^{0}P_{C}^{0}\ket{\varPsi}\ket{000}+P_{A}^{0}X_{B}P_{B}^{1}X_{C}P_{C}^{1}\ket{\varPsi}\ket{011}\nonumber\\
&+X_{A}P_{A}^{1}P_{B}^{0}X_{C}P_{C}^{1}\ket{\varPsi}\ket{101}+X_{A}P_{A}^{1}X_{B}P_{B}^{1}P_{C}^{0}\ket{\varPsi}\ket{110}\nonumber\\
&+X_{A}P_{A}^{1}X_{B}P_{B}^{1}X_{C}P_{C}^{1}\ket{\varPsi}\ket{111}\|\nonumber\\
\leq &\|P_{A}^{0}P_{B}^{0}P_{C}^{0}\ket{\varPsi}\ket{000}\|+\|P_{A}^{0}X_{B}P_{B}^{1}X_{C}P_{C}^{1}\ket{\varPsi}\ket{011}\|+\nonumber\\
&\|X_{A}P_{A}^{1}P_{B}^{0}X_{C}P_{C}^{1}\ket{\varPsi}\ket{101}\|+\|X_{A}P_{A}^{1}X_{B}P_{B}^{1}P_{C}^{0}\ket{\varPsi}\ket{110}\|\nonumber\\
&+\|X_{A}P_{A}^{1}X_{B}P_{B}^{1}X_{C}P_{C}^{1}\ket{\varPsi}\ket{111}\|\nonumber\\
=&|\bra{\varPsi}P_{A}^{0}P_{B}^{0}P_{C}^{0}\ket{\varPsi}|+|\bra{\varPsi}P_{A}^{0}P_{B}^{1}P_{C}^{1}\ket{\varPsi}|\nonumber\\
&+|\bra{\varPsi}P_{A}^{1}P_{B}^{0}P_{C}^{1}\ket{\varPsi}|+|\bra{\varPsi}P_{A}^{1}P_{B}^{1}P_{C}^{0}\ket{\varPsi}|\nonumber\\
&+|\bra{\varPsi}P_{A}^{1}P_{B}^{1}P_{C}^{1}\ket{\varPsi}|\nonumber\\
=&\varepsilon_{14}+\varepsilon_{15}+\varepsilon_{16}+\varepsilon_{17}+\varepsilon_{18}\nonumber\\
=&\delta_2. 
\end{align} 
Thus, we then obtain the distance between the state (\ref{general}) and (\ref{aim}) as
\begin{align}
\|\ket{\Psi'}-\ket{\tilde{\Psi}}\|\leq 3\delta_1+\delta_2.
\end{align}
The norm of $\ket{\tilde{\Psi}}$ can be estimated,
\begin{align}
&\|\ket{\tilde{\Psi}}\|^2=\|P_{A}^{0}P_{B}^{0}P_{C}^{0}X_{C}\ket{\varPsi}(\ket{001}+\ket{010}+\ket{100})\|^2 \nonumber\\
=& 3\|P_{A}^{0}P_{B}^{0}P_{C}^{0}X_{C}\ket{\varPsi}\|^2=3|\bra{\varPsi}X_C P_{A}^{0}P_{B}^{0}P_{C}^{0}X_{C}\ket{\varPsi}|\nonumber\\
=& 3|\bra{\varPsi}X_C P_{A}^{0}P_{B}^{0}(P_{C}^{0}X_{C}-X_{C}P_{C}^{1}+X_{C}P_{C}^{1})\ket{\varPsi}|\nonumber\\
\leq & 3|\bra{\varPsi}X_C P_{A}^{0}P_{B}^{0}(P_{C}^{0}X_{C}-X_{C}P_{C}^{1})\ket{\varPsi}|\nonumber \\
&+3|\bra{\varPsi}X_C P_{A}^{0}P_{B}^{0}X_{C}P_{C}^{1})\ket{\varPsi}|\nonumber\\
\leq & 3\|P_{B}^{0}X_C\ket{\varPsi}\|\cdot\|P_{A}^{0}(P_{C}^{0}X_{C}-X_{C}P_{C}^{1})\ket{\varPsi}\|\nonumber\\
&+3|\bra{\varPsi}X_C P_{A}^{0}P_{B}^{0}X_{C}P_{C}^{1})\ket{\varPsi}|\nonumber\\
\leq & 3\cdot\delta_1+3|\bra{\varPsi}X_C^2 P_{A}^{0}P_{B}^{0}P_{C}^{1})\ket{\varPsi}|\nonumber\\
=& 1-3\varepsilon_1+3\delta_1,
\end{align}
and,
\begin{align}
&\|\ket{\tilde{\Psi}}\|^2\nonumber\\
=& 3|\bra{\varPsi}X_C P_{A}^{0}P_{B}^{0}(P_{C}^{0}X_{C}-X_{C}P_{C}^{1}+X_{C}P_{C}^{1})\ket{\varPsi}|\nonumber\\
\geq & 3|\bra{\varPsi}X_C P_{A}^{0}P_{B}^{0}X_{C}P_{C}^{1})\ket{\varPsi}|\nonumber\\
&-3|\bra{\varPsi}X_C P_{A}^{0}P_{B}^{0}(P_{C}^{0}X_{C}-X_{C}P_{C}^{1})\ket{\varPsi}|\nonumber\\
\geq & 1-3\varepsilon_1-3\delta_1.
\end{align}
These results imply that,
\begin{align}
\|\ket{\tilde{\Psi}}-\ket{\textrm{junk}}_{ABC}\ket{W_{3}}_{A'B'C'}\| \leq 1-\sqrt{1-3\varepsilon_1-3\delta_1}
\end{align}
where,
\begin{align}
\ket{\textrm{junk}}_{ABC}\ket{W_{3}}_{A'B'C'}=\frac{\ket{\tilde{\Psi}}}{\|\ket{\tilde{\Psi}}\|}.
\end{align}
Finally,
\begin{align}
&\|\ket{\Psi'}-\ket{\textrm{junk}}_{ABC}\ket{W_{3}}_{A'B'C'}\|\nonumber\\
\leq & \|\ket{\Psi'}-\ket{\tilde{\Psi}}\|+\|\ket{\tilde{\Psi}}-\ket{\textrm{junk}}_{ABC}\ket{W_{3}}_{A'B'C'}\|\nonumber\\
= & 3\delta_1+\delta_2+1-\sqrt{1-3\varepsilon_1-3\delta_1}.
\end{align}
As we have said, without losing the generality, we take the maximum $\varepsilon$ among $\varepsilon_i$s for notational simplicity. Then the relaxed observation requirement will not affect the robustness bound proved below. So a conservative upper bound will be
\begin{align}
\|\ket{\Psi'}&-\ket{\textrm{junk}}_{ABC}\ket{W_{3}}_{A'B'C'}\|\nonumber\\
&\leq  \frac{13}{2}\varepsilon+9(2+2\sqrt{2})(\frac{20}{3})^{\frac{1}{4}}\frac{9\sqrt{15}+6\sqrt{5}}{20}\varepsilon^{\frac{3}{4}}\nonumber\\
&+9(2+2\sqrt{2})(\frac{20}{3})^{\frac{1}{4}}\varepsilon^{\frac{1}{4}}\nonumber\\
&\approx 7.5\varepsilon+119.2\varepsilon^{\frac{3}{4}}+49.4\varepsilon^{\frac{1}{4}}.
\end{align}

\section{Proof of \textbf{Lemma} \ref{lemma2}}
\label{general3qubitproof}
In principle, an isometry of the form described in Fig.\ref{fig:isometry}, can be constructed with: 

\begin{eqnarray}
Z_A = -A_0, X_A = A_1, \nonumber \\
Z_B = -B_0, X_B = B_1, \nonumber \\
Z_C = \frac{C_0+C_1}{2\cos\mu}, X_C = \frac{C_0-C_1}{2\sin\mu},
\end{eqnarray}
resulting in \eqref{general}. 

To evaluate this state, first note from $\average{P_A^1P_B^0 + P_A^0P_B^1 + P_A^0P_B^0} = 1$, that we must have $P_A^1P_B^1 = 0$. 

Second, following the self testing of non-maximally entangled qubits, maximal violation of the tilted Bell inequality $\beta$ implies,
\begin{eqnarray}
P_A^0P_C^1\ket{\Psi} = P_A^0P_B^0\ket{\Psi}, \nonumber \\
P_B^0P_C^1\ket{\Psi} = P_B^0P_A^0\ket{\Psi}, \nonumber \\
P_A^0P_C^0\ket{\Psi} = P_A^0P_B^1\ket{\Psi}, \nonumber \\
P_B^0P_C^0\ket{\Psi} = P_B^0P_A^1\ket{\Psi}, \nonumber \\
P_A^0X_CP_C^1\ket{\Psi} = P_A^0P_C^0X_C\ket{\Psi}, \nonumber \\
P_B^0X_CP_C^1\ket{\Psi} = P_B^0P_C^0X_C\ket{\Psi}, \nonumber \\
P_A^0X_CP_C^0\ket{\Psi}/\gamma = P_A^0X_BP_B^1\ket{\Psi}, \nonumber \\
P_B^0X_CP_C^0\ket{\Psi}/\gamma = P_B^0X_AP_A^1\ket{\Psi},
\end{eqnarray}
which implies that
\begin{eqnarray}
P_A^0P_B^0P_C^0 \ket{\Psi} = P_A^0P_C^1P_C^0\ket{\Psi} = 0, \nonumber \\
P_A^0P_B^1P_C^1 \ket{\Psi} = P_A^0P_C^0P_C^1\ket{\Psi} = 0, \nonumber \\
P_A^1P_B^0P_C^1 \ket{\Psi} = P_C^0P_B^0P_C^1\ket{\Psi} = 0,
\end{eqnarray}
and 
\begin{eqnarray}
P_A^0P_B^0X_CP_C^1 \ket{\Psi} &=& X_CP_C^1P_A^0P_B^0 \ket{\Psi} \nonumber \\
&=& X_CP_C^1P_A^0P_C^1 \ket{\Psi} \nonumber \\
&=& P_A^0 X_CP_C^1\ket{\Psi} \nonumber \\
&=& P_B^0 X_CP_C^1\ket{\Psi},
\end{eqnarray}
but also
\begin{eqnarray}
(P_A^0X_BP_B^1)P_C^0 \ket{\Psi} &=& P_C^0(1/\gamma P_A^0X_CP_C^1) \ket{\Psi}\nonumber \\
&=& 1/\gamma P_A^0P_C^0X_C \ket{\Psi} \nonumber \\ 
&=& 1/\gamma P_A^0X_CP_C^1 \ket{\Psi},
\end{eqnarray}
and similary,
\begin{eqnarray}
(P_B^0X_AP_A^1)P_C^0 \ket{\Psi} &=& P_C^0(1/\gamma P_B^0X_CP_C^1) \ket{\Psi}\nonumber \\
&=& 1/\gamma P_B^0P_C^0X_C \ket{\Psi}\nonumber \\ 
&=& 1/\gamma P_B^0X_CP_C^1 \ket{\Psi}.
\end{eqnarray}
Combining these relations gives,
\begin{eqnarray}
\Phi(\ket{\Psi}	) &=& \gamma P_A^0X_CP_C^1 \ket{\Psi} (\ket{001}+\ket{010}+\gamma \ket{001}) \nonumber \\
&=& \ket{\textrm{junk}} \otimes \ket{\psi_\gamma}
\end{eqnarray} 
as claimed.

\end{document}